\documentclass[11pt]{article}
\usepackage[letterpaper,margin=1in]{geometry}

\usepackage{times} 
\usepackage{epsfig,xspace} 
\usepackage{graphicx,color}
\usepackage{algorithm,algorithmicx,algpseudocode}
\usepackage{amsmath,amssymb,amsfonts,amsthm,mathrsfs}
\usepackage{url}
\usepackage{wrapfig}
\usepackage{times} 
\usepackage{authblk}

\newtheorem{theorem}{Theorem}
\newtheorem{lemma}{Lemma}

\newtheorem{definition}{Definition}

\newtheorem{proposition}{Proposition}

\newcommand{\pt}{{\mathcal T}}
\newcommand{\cost}{{\rm cost}}

\title{A Logarithmic Integrality Gap Bound for Directed Steiner Tree
  in Quasi-bipartite Graphs
  \footnote{This work was in part supported by NSERC's Discovery grant program.
    The second author greatfully acknowledges the support of the Hausdorff
    Institute and the Institute for Discrete Mathematics in Bonn, Germany.}
}

\author[1]{Zachary Friggstad}
\author[2]{Jochen K\"{o}nemann}
\author[3]{Mohammad Shadravan}
\affil[1]{Department of Computing Science, University of Alberta\\
  Edmonton, AB, Canada, T6G 2E8\\
  \texttt{zacharyf@ualberta.ca}}
\affil[2]{Department of Combinatorics and Optimization, University of Waterloo\\
  Waterloo, ON, Canada, N2L 3G1\\  \texttt{jochen@uwaterloo.ca}}
\affil[3]{Department of Industrial Engineering and Operations Research,
Columbia University\\
New York, NY, USA, 10027\\
  \texttt{ms4961@columbia.edu}}

\begin{document}

\maketitle

\begin{abstract}
  We demonstrate that the integrality gap of the natural cut-based LP
  relaxation for the directed Steiner tree problem is $O(\log k)$ in
  quasi-bipartite graphs with $k$ terminals. Such instances can be
  seen to generalize set cover, so the integrality gap analysis is
  tight up to a constant factor.  A novel aspect of our approach is that we
  use the primal-dual method; a technique that is rarely
  used in designing approximation algorithms for network design
  problems in directed graphs.
 \end{abstract}

\section{Introduction}

In an instance of the {\em directed Steiner tree} (DST) problem, we are
given a directed graph $G = (V,E)$, non-negative costs $c_e$ for all
$e \in E$, {\em terminal} nodes $X \subseteq V$, and a root $r \in
V$. The remaining nodes in $V - (X \cup \{r\})$ are the {\em Steiner nodes}.
The goal is to find the cheapest collection of edges $F \subseteq
E$ such that for every terminal $t \in X$ there is an $r,t$-path using only
edges in $F$. Throughout, we let $n$ denote $|V|$ and $k$ denote $|X|$.

If $X \cup \{r\} = V$, then the problem is simply the {\em
  minimum-cost arborescence} problem which can be solved
efficiently~\cite{edm67}. However, the general case is well-known to
be NP-hard. In fact, the problem can be seen to generalize the 
{\em set-cover} and 
{\em group
Steiner tree} problems. The latter cannot be approximated within
$O(\log^{2-\epsilon}(n))$ for any constant $\epsilon > 0$ unless ${\rm
  NP} \subseteq {\rm DTIME}(n^{{\rm polylog}(n)})$
\cite{halperin:krauthgamer}.

For a DST instance $G$, let $OPT_G$
denote the value of the optimum solution for this instance 
Say that an instance $G = (V,E)$ of DST with terminals $X$ is {\em
  $\ell$-layered} if $V$ can be partitioned as $V_0, V_1, \ldots,
V_\ell$ where $V_0 = \{r\}, V_\ell = X$ and every edge $uv \in E$ has
$u \in V_i$ and $v \in V_{i+1}$ for some $0 \leq i < \ell$.
Zelikovsky showed for any DST
instance $G$ and integer $\ell \geq 1$ that we can compute an $\ell$-layered
DST instance $H$ in poly$(n,\ell)$ time such that $OPT_G \leq OPT_H \leq \ell \cdot k^{1/\ell}
\cdot OPT_G$ and that a DST solution in $H$ can be efficiently mapped to a
DST solution in $G$ with the same cost~\cite{calinescu:zelikovsky,zelikovsky}.

Charikar et al.~\cite{charikar} exploited this fact and presented an $O(\ell^2\cdot k^{1/\ell} \cdot \log k)$-approximation
with running time ${\rm poly}(n, k^{\ell})$ for any integer $\ell \geq 1$.
In particular, this can be used to obtain an $O(\log^3 k)$-approximation in quasi-polynomial time
and a polynomial-time $O(k^\epsilon)$-approximation for any constant $\epsilon > 0$.
Finding a polynomial-time polylogarithmic approximation remains an important open problem.

For a set of nodes $S$, we let $\delta^{in}(S) = \{uv \in V : u \not\in S {\rm ~and~} v \in S\}$
be the set of edges entering $S$.
The following is 
a natural linear programming (LP) relaxation for directed Steiner tree.
\begin{alignat}{3}
\min & \quad & \sum_{e \in E} c_e& x_e  \tag{\bf DST-Primal} \label{lp:primal} \\
\text{s.t.} && x(\delta^{\rm in}(S)) \geq & ~1 && ~~\forall~ S \subseteq V-r, S \cap X \neq \emptyset \label{cnst:primal} \\
&& x_e \geq & ~0 && ~~\forall e \in E \notag
\end{alignat}
This LP is called a {\em relaxation} because of the natural correspondence between feasible solutions to a DST instance $G$
and feasible $\{0,1\}$-integer solutions to the corresponding LP (\ref{lp:primal}). Thus, if we let $OPT_{LP}$ denote the
value of an optimum (possibly fractional) solution to LP (\ref{lp:primal}) then we have $OPT_{LP} \leq OPT_G$.
For a particular instance $G$ we say the {\em integrality gap} is $OPT_G / OPT_{LP}$; we are interested in placing the smallest
possible upper bound on this quantity.

Interestingly, if $|X| = 1$ (the shortest path problem) or $X \cup \{r\} = V$ (the minimum-cost arborescence problem), the extreme points
of \eqref{lp:primal} are integral so the integrality gap is 1 (\cite{papadimitriou} and \cite{edm67}, respectively).
However, in the general case Zosin and Khuller showed that \eqref{lp:primal} is not useful 
for finding $\mbox{polylog}(k)$-approximation algorithms for DST~\cite{zosin:khuller}.
The authors showed that the integrality gap of \eqref{lp:primal} relaxation can,
unfortunately, be as bad as $\Omega(\sqrt k)$, even in instances where
$G$ is a 4-layered graph. In their examples, the number of nodes $n$
is exponential in $k$ so the integrality gap may still be $O(\log^c n)$ for some constant $c$.

On the other hand, Rothvoss
recently showed that applying $O(l)$ rounds of the semidefinite programming
Lasserre hierarchy to (the flow-based extended formulation of)
\eqref{lp:primal} yields an SDP with integrality gap $O(\ell\cdot\log k)$
for $\ell$-layered instances~\cite{rothvoss}. Subsequently, Friggstad et al.~\cite{friggstad+14}
showed similar results for the weaker Sherali-Adams and Lov\'asz-Schrijver
linear programming hierarchies.

In this paper we consider the class of {\em quasi-bipartite} DST
instances. An instance of DST is quasi-bipartite if the Steiner
nodes $V\setminus (X \cup \{r\})$ form an independent set (i.e., no
directed edge has both endpoints in $V\setminus (X \cup \{r\})$).
Such instances still capture the set cover problem, and thus do not
admit an $(1-\epsilon)\ln k$-approximation for any constant
$\epsilon > 0$ unless ${\rm P} = {\rm NP}$~\cite{steurer,Fe98}.
Furthermore, it is straightforward to adapt known integrality gap
constructions for set cover (e.g.~\cite{vazirani}) 
to show that
the integrality gap of \eqref{lp:primal} can be as bad as
$(1 - o(1)) \cdot \ln k$ in some instances.  Hibi and
Fujito~\cite{HF12} give an $O(\log k)$-approximation for
quasi-bipartite instances of DST, but do not provide any integrality
gap bounds.

Quasi-bipartite instances have been well-studied in the context of {\em
undirected} Steiner trees.  The class of graphs was first
introduced by Rajagopalan and Vazirani~\cite{rv99} who studied the
integrality gap of \eqref{lp:primal} for the {\em bidirected} map of the
given undirected Steiner tree instances. Currently, the best approximation
for quasi-bipartite instances of undirected Steiner tree is $\frac{73}{60}$
by Goemans et al. \cite{goemans} who also bound the integrality gap
of the bidirected cut relaxation by the same quantity. This is the same LP relaxation as
(\ref{lp:primal}), applied to the directed graph obtained by replacing each undirected
edge $\{u,v\}$ with the two directed edges $uv$ and $vu$.
This is a slight improvement over a prior
$(\frac{73}{60} + \epsilon)$-approximation for any constant $\epsilon > 0$
by Byrka et al. \cite{byrka}.

The best approximation for general instances
of undirected Steiner tree is $\ln(4) + \epsilon$ for any constant $\epsilon > 0$
\cite{byrka}. However, the best known upper bound on the integrality gap of the bidirected cut relaxation
for non-quasi-bipartite instances is only 2; it is an open problem to determine
if this integrality gap is a constant-factor better than 2.

\subsection{Our contributions}
Our main result is the following. Let $H_n = \sum_{i=1}^n 1/i = O(\log n)$ be the $n$th harmonic number.

\begin{theorem}\label{thm:qbt-gap}
  The integrality gap of LP (\ref{lp:primal}) is at most $2 \, H_k
  = O(\log k)$ in quasi-bipartite graphs with $k$ terminals.
  Furthermore, a Steiner tree with cost at most $2 \,H_k \cdot OPT_{LP}$ can be constructed in polynomial time.
\end{theorem} 

As noted above, 
Theorem \ref{thm:qbt-gap} is asymptotically tight since any of the well-known $\Omega(\log k)$ integrality
gap constructions for set cover instances with $k$ items translate directly to an integrality gap
lower bound for (\ref{lp:primal}), using the usual reduction from set cover to 2-layered
quasi-bipartite instances of directed Steiner tree.

This integrality gap bound asymptotically matches the approximation guarantee
proven by Hibi and Fujito for quasi-bipartite DST instaces~\cite{HF12}.
We remark that their approach is unlikely to give any integrality gap bounds for (\ref{lp:primal}) because they iteratively choose {\em low-density full Steiner trees}
in the same spirit as~\cite{charikar} and give an $O(\ell \cdot \log k)$-approximation for finding the optimum DST solution $T$ that does not contain
a path with $\geq \ell$ Steiner nodes $V\setminus (X \cup \{r\})$. In particular, their approach will also find an $O(\log k)$-approximation
to the optimum DST solution in 4-layered graphs and we know the integrality gap in some 4-layered instances is $\Omega(\sqrt k)$~\cite{zosin:khuller}.

We prove Theorem \ref{thm:qbt-gap} by constructing a directed Steiner tree
in an iterative manner. An iteration starts with a {\em partial}
Steiner tree (see Definition \ref{def:partial} below), which consists of multiple directed components containing
the terminals in $X$. Then a set of arcs are purchased to {\em augment} this partial solution
to one with fewer directed components. These arcs are discovered through a primal-dual moat growing procedure;
a feasible solution for the dual \eqref{lp:primal} is constructed and the cost of the purchased arcs can be bounded
using this dual solution.

While the primal-dual technique has been very successful for {\em
  undirected} network design problems (e.g., see \cite{forest}), far
fewer success stories are known in {\em directed} domains. Examples
include a primal-dual interpretation of Dijkstra's shortest path
algorithm (e.g., see Chapter 5.4 of \cite{papadimitriou}), and
Edmonds'~\cite{edm67} algorithm for minimum-cost arborescences. In
both cases, the special structure of the problem is instrumental in
the primal-dual construction. One issue arising in the implementation
of primal-dual approaches for directed network design problems appears
to be a certain {\em overlap} in the {\em moat} structure maintained by
these algorithms. We are able to handle this difficulty here by exploiting the
quasi-bipartite nature of our instances.  


\section{The integrality gap bound} \label{sec:rounding}

\subsection{Preliminaries and definitions}
We now present an algorithmic proof of Theorem \ref{thm:qbt-gap}.
As we will follow a primal-dual strategy, we first present the LP dual of
\eqref{lp:primal}. 

\begin{alignat}{3}
\max & \quad & \sum_S y_S &  \tag{\bf DST-Dual} \label{lp:dual} \\
\text{s.t.} && \sum_{S : e \in \delta^{\rm in}(S)} y_S \leq & ~c_e && ~\forall e \in E  \label{cnst:dual} \\
&& y \geq & 0 && \nonumber
\end{alignat}
In (\ref{lp:dual}), the sums range only over sets of nodes $S$ such that $S \subseteq V-r$ and $S \cap X \neq \emptyset$.

Our algorithm builds up partial solutions, which are defined as follows.
\begin{definition} \label{def:partial}
A {\em partial Steiner tree} is a tuple $\pt = (\{B_i, h_i, F_i\}_{i=0}^\ell, \bar{B})$ where, for each $0 \leq i \leq \ell$,
$B_i$ is a subset of nodes, $h_i \in B_i$,
and $F_i$ is a subset of edges with endpoints only in $B_i$ such that the following hold.
\begin{itemize}
\item The sets $B_0, B_1, \ldots, B_\ell, \bar{B}$ form a partition $V$.
\item $\bar{B} \subseteq V-X-r$ (i.e. $\bar{B}$ is a subset of Steiner nodes).
\item $h_0 = r$ and $h_i \in X$ for each $1 \leq i \leq \ell$.
\item For every $0 \leq i \leq \ell$ and every $v \in B_i$, $F_i$ contains an $h_i,v$-path.
\end{itemize}
We say that $\bar{B}$ is the set of {\em free Steiner nodes} in $\pt$ and that $h_i$ is the {\em head} of $B_i$ for
each $0 \leq i \leq \ell$.
The edges of $\pt$, denoted by $E(\pt)$, are simply $\cup_{i=0}^\ell F_i$.
We say that $B_0, \ldots, B_\ell$ are the {\em components} of $\pt$ where $B_0$ is the {\em root component}
and $B_1, \ldots, B_\ell$ are the {\em non-root components}.
\end{definition}
Figure \ref{fig:partial} illustrates a partial Steiner tree.
Note that if $\pt$ is a partial Steiner tree with $\ell = 0$ non-root components, then $E(\pt)$ is in fact a feasible DST solution.

Finally, for a subset of edges $F$ we let $\cost(F) = \sum_{e \in F} c_e$.

\begin{figure}[ht]
\centering
\includegraphics[scale=0.7]{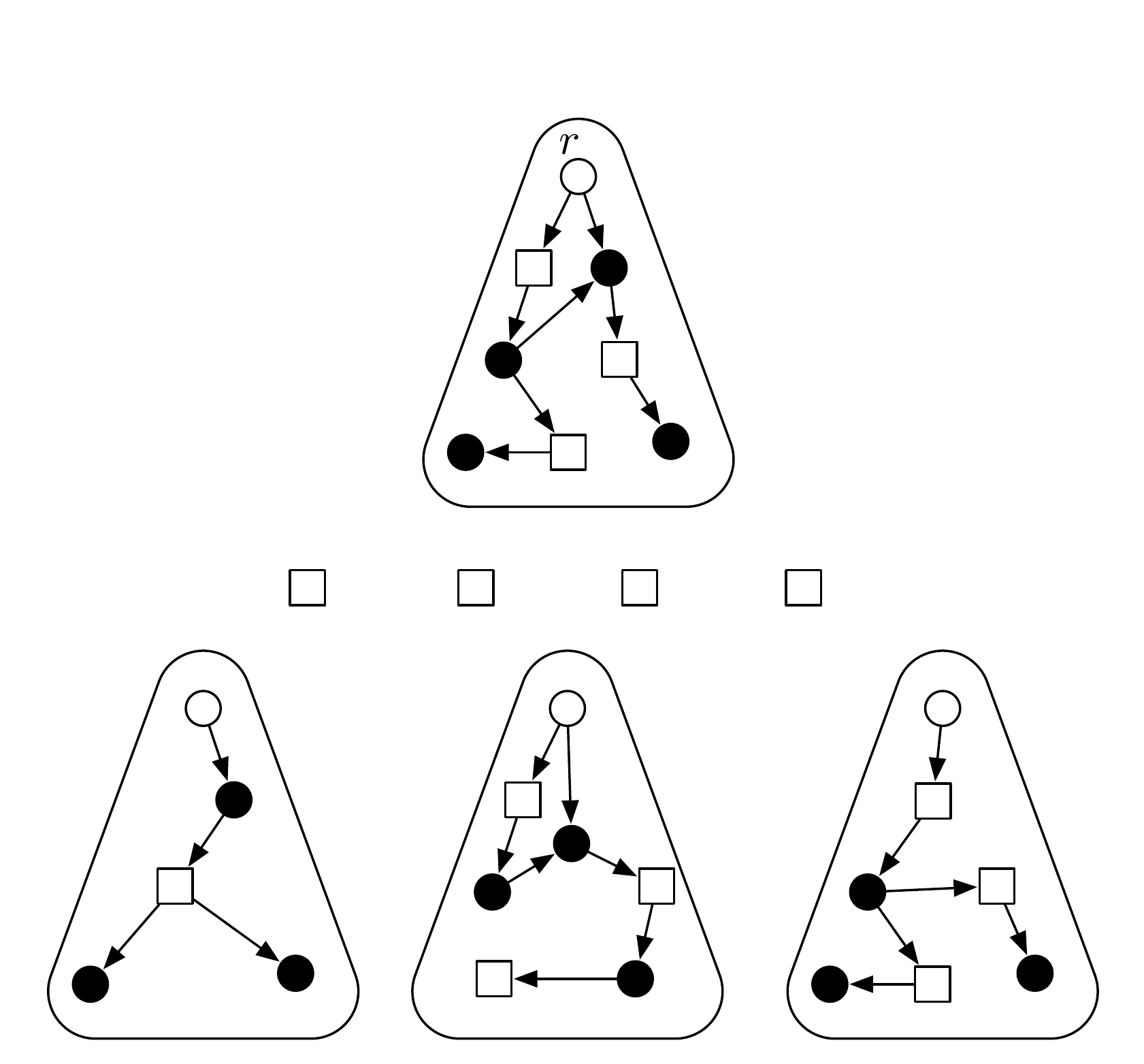}
\caption{A partial Steiner tree with $\ell = 3$ non-root components (the root is pictured at the top). The only edges shown are those
in some $F_i$.
The white circles are the heads of the various sets $B_i$ and the black circles are terminals that are not heads
of any components. The squares outside of the components are the free Steiner nodes $\bar{B}$. Note, in particular, that each head can reach every
node in its respective component. We do not require each $F_i$ to be a minimal set of edges with this property. }
\label{fig:partial}
\end{figure}

\subsection{High-level approach}

Our algorithm builds up partial Steiner trees in an iterative manner
while ensuring that the cost does not increase by a
significant amount between iterations.  Specifically, we prove the
following lemma in Section \ref{sec:growing}. Recall that $OPT_{LP}$
refers to the optimum solution value for (\ref{lp:primal}).
\begin{lemma}\label{lem:grow}
Given a partial Steiner tree $\pt$ with $\ell \geq 1$ non-root components, there is a polynomial-time algorithm
that finds a partial Steiner tree $\pt'$ with $\ell' < \ell$ non-root components such that
\[ \cost(E(\pt')) \leq \cost(E(\pt)) + 2 \cdot OPT_{LP}\cdot \frac{\ell - \ell'}{\ell}.\]
\end{lemma}

Theorem \ref{thm:qbt-gap} follows from Lemma \ref{lem:grow} in a standard way.

\begin{proof}[Proof of Theorem \ref{thm:qbt-gap}]
  Initialize a partial Steiner tree $\pt_k$ with $k$ non-root
  components as follows. Let $\bar{B}$ be the set of all Steiner nodes,
  $B_0=\{r\}$, and $F_0=\emptyset$. Furthermore, label
  the terminals as $t_1, \ldots, t_k \in X$ and for each
  $1 \leq i \leq k$ let $B_i = \{t_i\}, h_i = t_i$ and
  $F_i = \emptyset$. Note that $\cost(E(\pt_k)) = 0$.

Iterate Lemma \ref{lem:grow} to obtain a sequence of partial Steiner trees $\pt_{\ell_0}, \pt_{\ell_1}, \pt_{\ell_2}, \ldots, \pt_{\ell_a}$
where $\pt_{\ell_i}$ has $\ell_i$ non-root components such that $k = \ell_0 > \ell_1 > \ldots > \ell_a = 0$ and
\[\cost(E(\pt_{i+1})) \leq \cost(E(\pt_i)) + 2 \cdot OPT_{LP}\cdot \frac{\ell_i - \ell_{i+1}}{\ell_i}\] for each $0 \leq i < a$. Return $E(\pt_{\ell_a})$
as the final Steiner tree. 

That $E(\pt_a)$ can be found efficiently follows simply because we are iterating the efficient algorithm from Lemma \ref{lem:grow}
at most $k$ times. The cost of this Steiner tree can be bounded as follows.
\begin{eqnarray*}
\cost(E(\pt_a)) & \leq & 2\cdot OPT_{LP} \cdot \sum_{i=0}^{a-1} \frac{\ell_i - \ell_{i+1}}{\ell_i}
= 2\cdot OPT_{LP}\cdot\sum_{i=0}^{a-1} ~\sum_{j=\ell_{i+1}+1}^{\ell_i} \frac{1}{\ell_i} \\
& \leq & 2\cdot OPT_{LP}\cdot\sum_{i=0}^{a-1} ~\sum_{j=\ell_{i+1}+1}^{\ell_i} \frac{1}{j} = 2\cdot OPT_{LP}\cdot \sum_{j=1}^k \frac{1}{k} \\
& = & 2\cdot OPT_{LP}\cdot H_k.
\end{eqnarray*}
\end{proof}

The idea presented above resembles one proposed by Guha et
al.~\cite{guha} for bounding the integrality gap of a natural
relaxation for {\em undirected} node-weighted Steiner tree by
$O(\log k)$~\cite{guha}. Like our approach, Guha et al.~also build a
solution incrementally. In each {\em phase} of the algorithm, the
authors reduce the number of connected components of a partial
solution by adding vertices whose cost is charged carefully to the
value of a dual LP solution that the algorithm constructs
simultaneously.


\section{A primal-dual proof of Lemma \ref{lem:grow}} \label{sec:growing}

Consider a given partial Steiner tree
$\pt = (\{B_i, h_i, F_i\}_{i=0}^\ell, \bar{B})$ with $\ell \geq 1$
non-root components. Lemma \ref{lem:grow} promises a partial Steiner
tree $\pt'$ with $\ell' < \ell$ non-root components with
$\cost(E(\pt')) \leq \cost(E(\pt)) + 2 \cdot OPT_{LP} \cdot \frac{\ell
  - \ell'}{\ell}$.
In this section we will present an algorithm that {\em augments}
forest $\pt$ in the sense that it computes a set of edges to {\em add}
to $\pt$. The proof presented here is constructive: we will design a
{\em primal-dual} algorithm that maintains a feasible dual solution
for \eqref{lp:dual}, and uses the structure of this solution to
guide the process of adding edges to $\pt$.

\subsection{The algorithm}
For any two nodes $u,v \in V$, let $d(u,v)$ be the cost of the
cheapest $u,v$-path in $G$.  More
generally, for a subset $\emptyset \subsetneq S \subseteq V$ and a
node $v \in V$ we let $d(S,v) = \min_{u \in S} d(u,v)$.  We will
assume that for every $0 \leq i \leq \ell$ and
$1 \leq j \leq \ell, j \neq i$ that $d(B_i,h_j) > 0$ as otherwise, we
could merge $B_i$ and $B_j$ by adding the $0$-cost $B_i,h_j$-path to
$\pt$. 

The usual conventions of primal-dual algorithms will be adopted. We
think of such an algorithm as a continuous process that increases the
value of some dual variables over time. At time $t=0$, all dual
variables are initialized to a value of $0$. At any point in time,
exactly $\ell$ dual variables will be raised at a rate of one unit per
time unit. We will use $\Delta$ for the time at which the algorithm
terminates.  As is customary, we will say that an edge $e$ {\em goes
  tight} if the dual constraint for $e$ becomes tight as the dual
variables are being increased. When an edge goes tight, we will
perform some updates to the various sets being maintained by the
algorithm. Again, the standard convention applies that if multiple
edges go tight at the same time, then we process them in any order.

Algorithm \ref{alg:grow} describes the main subroutine that augments
the partial Steiner tree $\pt$ to one with fewer components.  It
maintains a collection of {\em moats} $M_i \subseteq V - \{r\}$ and
edges $F'_i$ for each $1 \leq i \leq \ell$, while ensuring that the
dual solution $y$ it grows remains feasible.  Mainly to aid notation,
our algorithm will maintain a so called {\em virtual body} $\beta_i$
for all $0 \leq i \leq \ell$ such that 
$B_i \subseteq \beta_i \subseteq B_i \cup \bar{B}$.
We will ensure that each $v \in \bar{B} \cap \beta_i$ has a {\em mate} $u \in B_i$ such that the edge
$uv$ has cost no more than $\Delta$. For notational convenience, we
will let $\beta_0=B_0$ be the virtual body of the root component. The
algorithm will not grow a moat around the root since dual variables do not exist for sets containing the root.

Our algorithm will ensure that moats are pairwise {\em
terminal-disjoint}. In fact, we ensure that any two moats may only intersect in
$\bar{B}$. Terminal-disjointness together with the quasi-bipartite
structure of the input graph will allow us to charge the cost of arcs
added in the augmentation process to the duals grown.

An intuitive overview of our process is the following. At any time
$t \geq 0$, the moats $M_i$ will consist of all nodes $v$ with
$d(v,h_i) \leq t$. The moats $M_i$ will be grown until, at some time
$\Delta$, for at least one pair $i, j$ with $i\neq j$, there is a
tight path connecting $\beta_j$ to $h_i$. At this point the algorithm
stops, and adds a carefully chosen collection of tight arcs to the
partial Steiner tree that merges $B_j$ and $B_i$ (and potentially
other components). Crucially, the cost of the added arcs will be {\em
charged} to the value of the dual solution grown around the merged
components.

Due the structure of quasi-bipartite graphs, we are able to ensure that
in each step of the algorithm the active moats pay for at most one arc that is
ultimately bought to form $\pt'$. Also, if $\pt'$ has $\ell' < \ell$ non-root components
then each arc was paid for by moats around at most $\ell - \ell' + 1 \leq 2(\ell - \ell')$ different heads. So, the total cost of all purchased arcs
is at most $2(\ell - \ell') \cdot \Delta$.
Finally, the total dual grown is $\ell \cdot \Delta$, which is $\leq OPT_{LP}$ due to feasibility, so the cost of the edges bought can be bounded by $2 \frac{\ell - \ell'}{\ell} \cdot OPT_{LP}$.

\subsection{Algorithm and invariants}
Now we will be more precise. The primal-dual procedure is presented in Algorithm \ref{alg:grow}.
The following invariants will be
maintained at any time $0 \leq t \leq \Delta$ during the execution of Algorithm \ref{alg:grow}.

\begin{enumerate}
\item For each $1 \leq i \leq \ell$, $h_i \in M_i$ and $M_i \subseteq V-\{r\}$ (so there is a variable $y_{M_i}$ in the dual).
\item For each $1 \leq i \leq \ell$, $M_i = \{v \in V : d(v,h_i) < t\} \cup S$ where $S \subseteq \{v \in V : d(v, h_i) = t\}$.
\item $M_i \cap M_j \subseteq \bar{B}$ and both $\beta_i \cap \beta_j = M_i \cap \beta_j = \emptyset$ for distinct $0 \leq i,j \leq \ell$. 
\item For each $1 \leq i \leq \ell$ we have $B_i \subseteq \beta_i \subseteq B_i \cup \bar{B}$. Furthermore, for each each $v \in \beta_i - B_i$
there is a {\em mate} $u \in B_i$ such that $uv \in E$ and $c_{uv} \leq t$.
\item $y$ is feasible for LP (\ref{lp:dual}) with value exactly $\ell \cdot t$.
\end{enumerate}

These concepts are illustrated in Figure \ref{fig:pd_process}.

\begin{figure}[ht]
\centering
\includegraphics[scale=0.7]{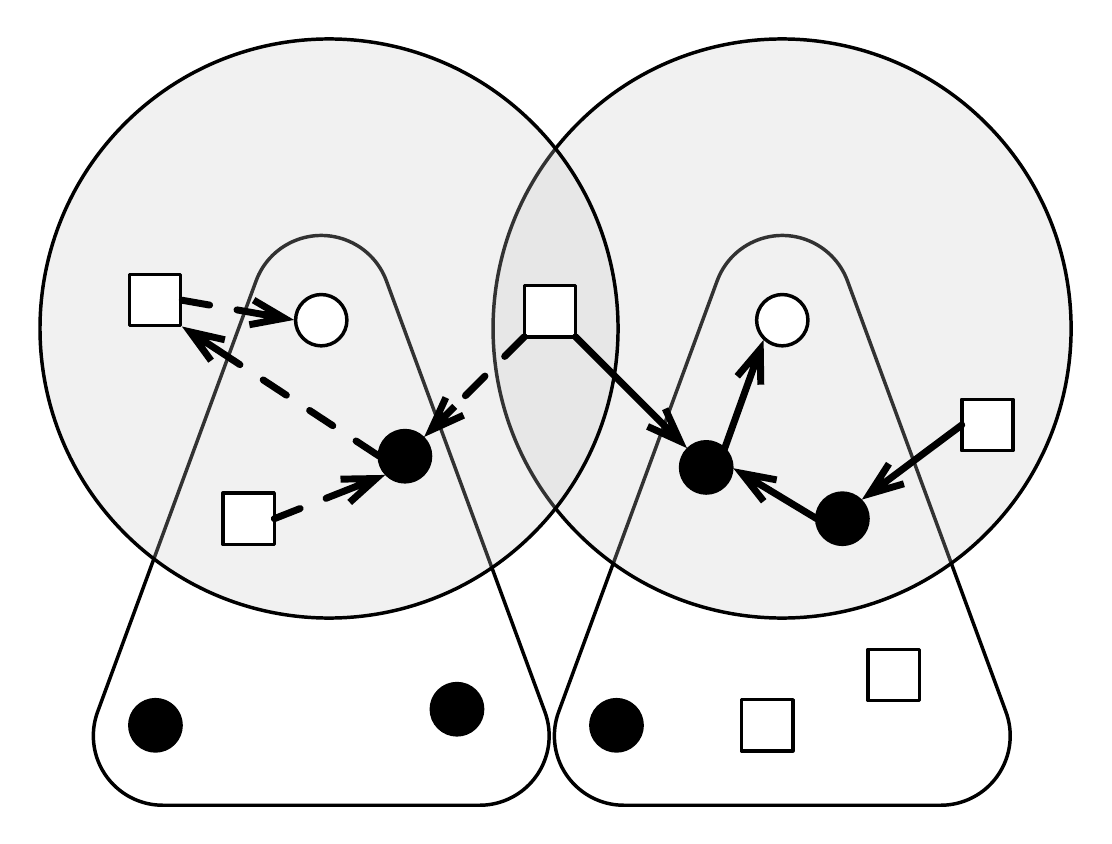}
\caption{ The moats around the two partial Steiner trees are depicted
  by the gray circles.  The dashed edges are those bought by the first
  moat and the solid edges are those bought by the second moat.  Note
  the moats only intersect in $\overline B$ (in particular, $v$ is the
  only lying in both moats).  Also, $u$ lies in the virtual body for
  the left partial Steiner tree and the dashed arc entering $u$ is
  coming from its mate.
  The edges $F_i$ from the original partial Steiner trees are not shown.
  Observe that if any edge entering $v$ goes tight then it must be
  from either $r$ or some terminal (because $G$ is quasi-bipartite).
  This would allow us to merge at least one partial Steiner tree into
  the body of another.
\label{fig:pd_process}
}
\end{figure}

\begin{algorithm*}[ht]
  \caption{~\bf{Dual Growing Procedure}}
  \label{alg:grow} 
\begin{algorithmic}[1] 
\State $M_i \leftarrow \{v \in V : d(v,h_i) = 0\}, 1 \leq i \leq \ell$
\State $\beta_i \leftarrow B_i$ for $0 \leq i \leq \ell$
\State $y \leftarrow {\bf 0}$

\State Raise $y_{M_{i'}}$ uniformly for each moat $M_{i'}$ until some edge $uv$ goes tight  \label{step:grow}

\If{$u \in \beta_j$ for some $0 \leq j \leq \ell$ {\bf and} $v \in M_{i'}$ for some $i' \neq j$} \label{step:bunch}
\State \Return the partial Steiner tree $\pt'$ described in Lemma \ref{lem:augment}. \label{step:stop}
\Else
\State Let $M_i$ be the unique moat with $uv \in \delta^{in}(M_i)$ \Comment c.f. Proposition \ref{prop:unique} \label{step:unique}
\State $M_i \leftarrow M_i \cup \{u\}$  \label{step:upa}
\If{$u \in \beta_i$}
\State $\beta_i \leftarrow \beta_i \cup \{v\}$    \label{step:upb}
\algtext*{EndIf} \EndIf
\State {\bf go to} Step \eqref{step:grow}
\algtext*{EndIf} \EndIf
\end{algorithmic}
\end{algorithm*}

\subsection{Invariant analysis}
\begin{lemma} \label{lem:invariants} Invariants 1-5 are maintained by
  Algorithm \ref{alg:grow} until the condition in the {\bf if}
  statement in Step \eqref{step:bunch} is true. Furthermore, the
  algorithm terminates in $O(n \cdot k)$ iterations.
\end{lemma}
\begin{proof}
  Clearly the invariants are true after the initialization steps (at
  time $t = 0$), given that $d(B_i,h_j) > 0$ for any $i \neq j$.
  To see why
  Algorithm \ref{alg:grow} terminates in a polynomial number of iterations, note that each iteration
  increases the size some moat by 1 and does not decrease the size of
  any moats. So after at most $kn$ iterations some moat will grow to include the virtual body of another moat,
  at which point the algorithm stops.
  
  Assume now that the invariants are true at some point just before
  Step \eqref{step:grow} is executed and that the condition in Step
  \eqref{step:bunch} is false after Step \eqref{step:grow} finishes.
  We will show that the invariants continue to hold just before the
  next iteration starts. We let $uv$ denote the edge that went tight that
  is considered in Step \eqref{step:grow}.
  We also let $t$ denote the total time the algorithm has executed
  (i.e. grown moats) up to this point. 
  
  Before proceeding with our proof, we exhibit the
  following useful fact. In what follows, let $M^{t'}_j$ be the moat around $h_j$ at any time $t' \leq t$ during
  the algorithm. This proposition demonstrates how we control the
  overlap of the moats by exploiting the quasi-bipartite structure.

  \begin{proposition} \label{prop:unique}
  If $uv \in \delta^{in}(M_i)$, 
    then $uv \not\in \delta^{in}(M^{t'}_j)$ for any $j \neq i$,
    and for any $t' \leq t$. 
  \end{proposition}
  \begin{proof}
    Suppose, for the sake of contradiction, that
    $uv \in \delta^{in}(M^{t'}_j)$ for some $j \neq i$ and
    $t' \leq t$.  Since $M^{t'}_j$ is a subset of $M_j$, the moat
    containing $h_j$ at time $t$, we must have $v \in M_j \cap
    M_i$.
    Invariant 3 now implies that $v \in \bar{B}$. Since $G$ is
    quasi-bipartite, then $u \in X$. Therefore
    $u \in B_{j'}$ for some $j'$. Since $j' \neq i$ or $j' \neq j$,
    then the terminating condition in Step \eqref{step:bunch} would have been satisfied as
    $u \in \beta_{j'}$. A contradiction.
  \end{proof}

Following Proposition \ref{prop:unique}, we let $i$ be the
  unique index such that $uv \in \delta^{in}(M_i)$ as in Step
  \eqref{step:unique}.

~

\noindent
  {\bf Invariant 1}\\ First note that $M_i$ never loses vertices during the algorithm's
  execution, and it therefore always contains head vertex $h_i$. 
  Also, vertex $u$ is not part of $B_0$ as otherwise the algorithm
  would have terminated in Step \eqref{step:bunch}. Hence $M_i \cup
  \{u\}$ also does not contain the root node $r$.

~

\noindent
  {\bf Invariant 2}\\
  This is just a reinterpretation of Dijkstra's algorithm in the primal-dual framework (e.g. Chapter 5.4 of \cite{papadimitriou}),
  coupled with the fact that no edge considered in Step \eqref{step:grow} in some iteration
  crosses more than one moat at any given time (Proposition \ref{prop:unique}).
  
~

\noindent
  {\bf Invariant 3}\\  Suppose
  $(M_i \cup \{u\}) \cap M_j \not\subseteq \bar{B}$ for some
  $i \neq j$. This implies $u \in M_j \setminus \bar{B}$, and
  hence $u \in B_j \subseteq \beta_j$. Thus, the termination condition
  in Step \eqref{step:bunch} was satisfied and the algorithm should have
  terminated; contradiction.

  If $v$ is not added to $\beta_i$, and thus $\beta_i$ remains
  unchanged, $\beta_i \cap \beta_j = M_j \cap \beta_i = \emptyset$
  continues to hold for $j \neq i$. We also must have that
  $(M_i \cup \{u\}) \cap \beta_j = \emptyset$ for $i \neq j$, as
  otherwise $u \in \beta_j$ and this would violate the termination
  condition in Step \eqref{step:bunch}.

  Now suppose that $v$ is added to $\beta_i$. Then for $j \neq i$ we
  still have $(\beta_i \cup \{v\}) \cap \beta_j = \emptyset$ as
  otherwise $v \in \beta_j$ which contradicts $v \in M_i$ and the fact
  that Invariant 3 holds at the start of this iteration. We also have
  that $M_j \cap (\beta_i \cup \{v\})=\emptyset$ as otherwise
  $v \in M_j$. But this would mean that $u \in M_j$ as well by
  Proposition \ref{prop:unique}.
  We established above that $(M_i \cup \{u\}) \cap M_j \subseteq \overline B$.
  However, $\{u,v\} \subseteq (M_i \cup \{u\}) \cap M_j \subseteq \overline B$
  contradicts the fact that $G$ is quasi-bipartite.
  
~

\noindent
  {\bf Invariant 4}\\  That $B_i \subseteq \beta_i$ is clear simply
  because we only add nodes to the sets $\beta_i$.  Suppose now that
  $v$ is added to $\beta_i$. In this case, $v \not\in B_i$ as
  $B_i \subseteq \beta_i$ from the start. We claim that $v$ can also
  not be part of $B_j$ for some $j \neq i$, since otherwise
  $\emptyset \neq B_j \cap M_i \subseteq \beta_j \cap M_i$,
  contradicting Invariant 3. Hence $v \in \bar{B}$.  Note that the
  quasi-bipartiteness of $G$ implies that $u \in X$, and hence
  $u \in B_i$. Proposition \ref{prop:unique} finally implies
  that only the moats crossed by $uv$ are moats around $i$,
  so since the algorithm only grows one moat around $i$ at any
  time we have $c_{uv} \leq t$, and this completes the proof of Invariant 4.

~

\noindent
  {\bf Invariant 5}\\  The Step \eqref{step:grow} stops the first time a
  constraint becomes tight, so feasibility is maintained. In each
  step, the algorithm grows precisely $\ell$ moats
  simultaneously. Because the objective function of \eqref{lp:dual} is
  simply the sum of the dual variables, then the value of the dual is
  just $\ell$ times the total time spent growing dual variables.
\end{proof}

\subsection{Augmenting $\pt$}\label{sec:augment}

To complete the final detail in the description of the algorithm, we
now show how to  construct the partial Steiner tree after Step
\eqref{step:bunch} has been reached.  Lemma \ref{lem:invariants} shows
that Invariants 1 through 5 hold just before Step \eqref{step:grow} in the final iteration.
Say the final iteration executes for $\delta$ time units and that $uv$ is the edge that goes
tight and was considered in Step \eqref{step:bunch}.

\begin{lemma}\label{lem:augment}
When Step \eqref{step:stop} is reached in Algorithm \ref{alg:grow}, we can efficiently find
a partial Steiner tree $\pt'$ with $\ell' < \ell$ non-root components such that
$cost(E(\pt')) \leq cost(E(pt)) + 2\frac{\ell - \ell'}{\ell} \cdot OPT_{LP}$.
\end{lemma}

\begin{proof}
Let $j$ be the unique index such that $u \in \beta_j$ at time
$\Delta$.  There is exactly one such $j$ because
$\beta_i \cap \beta_j = \emptyset$ for $i \neq j$ is ensured by the invariants. Next, let
$J = \{i' \neq j : v \in M_{i'}\}$ and note that $J$ consists of all
indices $i'$ (except, perhaps, $j$) such that
$uv \in \delta^{in}(M_{i'})$.  By the termination condition,
$J \neq \emptyset$. Vertex $u$ lies in $\beta_j$ by definition. If
$u \not\in B_j$ then we let $w$ be the mate of $u$ as defined in
Invariant 4. Otherwise, if $w \in B_j$, we let $w=u$.

For notational convenience,
we let $P_j$ be the path consisting of the single edge $wu$ (or just the trivial path with no edges if $w = u$). 
In either case, say cost of $P_j$ is $\Delta - \epsilon_j$ where $\epsilon_j \geq 0$ (c.f. Invariant 4).
For each $i' \in J$, let $P_{i'}$ be a shortest
$v,h_{i'}$-path. Invariant 2 implies that
\begin{equation}\label{costpi}
c(P_{i'}) = \Delta - \epsilon_{i'},
\end{equation}
for some $\epsilon_{i'} \geq 0$.
Observe also that the tightness of $uv$ at time $\Delta$ and the
definition of $J$ imply that
\begin{equation}\label{uvtight}
 \sum_{i' \in J \cup \{j\}} \epsilon_{i'} \geq c_{uv}. 
\end{equation}
In fact, precisely a $\epsilon_{i'}$-value of the dual variables for $i' \neq j$ contribute to $c_{uv}$; the contribution of $j$'s variables
to $c_{uv}$ is at most $\epsilon_{i'}$.

Construct a partial Steiner tree $\pt'$ obtained from $\pt$ and Algorithm \ref{alg:grow} as follows.
\begin{itemize}
\item The sets $B_{j'}, F_{j'}$ and head $h_{j'}$ are unchanged for all $j' \not\in J \cup\{j\}$.
\item Replace the components $\{B_{i'}\}_{i' \in J \cup \{j\}}$ with a component $B := \bigcup_{i' \in J \cup \{j\}} \left(B_{i'} \cup V(P_{i'})\right)$ having
head $h := h_j$. The edges of this component in $\pt'$ are $F := \bigcup_{i' \in J \cup \{j\}} (F_{i'} \cup E(P_{i'})) \cup \{uv\}$.
\item The free Steiner nodes $\bar{B}'$ of $\pt'$ are the Steiner nodes not contained in any of these components.
\end{itemize}
Namely, $\bar{B}'$ consists of those nodes in $\bar{B}$ that are not contained on any path $P_{i'}, i' \in J \cup \{j\}$.

  We show that Steiner tree $\pt'$ as constructed above satisfies the
  conditions stated in Lemma \ref{lem:grow}.
  We first verify that $\pt'$ as constructed above is indeed a valid
  partial Steiner tree. Clearly the new sets
  $\bar{B}', \{B_i\}_{i \not\in J + j}$ and $B$ partition $V$ and
  $\bar{B}'$ is a subset of Steiner nodes.

Note that if $0 \in J \cup \{j\}$ in the above construction, then $j = 0$ because no moat contains $r$.
Thus, if $B_0$ is replaced when $B$ is constructed, then $r$ is the head of this new component.

Next, consider any $b \in B$. If $b \in B_j$ then there is an
$h_j,b$-path in $F_j \subseteq F$. If $b \in B_{i'}, i' \neq j$ then it
can be reached from $h_j$ in $(B,F)$ as follows. Follow the $h_j,w$-path in $F_j$,
then the $w,u$ path $P_j$, cross the edge $uv$, follow $P_{i'}$ to reach $h_{i'}$,
and finally follow the $h_{i'},b$-path in $F_{i'}$.
Finally, if $b \not\in B_{i'}$ for any $i' \in J + j$
then $b$ lies on some path $P_{i'}$, in which case it can be reached
in a similar way.

It is also clear that $E(\pt) \subseteq E(\pt')$ and that the number
of non-root components in $\pt'$ is $\ell-|J| < \ell$. Also,
$\cost(E(\pt')) - \cost(E(\pt))$ is at most the cost of the the paths
$\{P_{i'}\}_{i' \in J + i}$ plus $c_{uv}$.

It now easily follows from \eqref{costpi} and \eqref{uvtight} that 
\[
  \sum_{i' \in J \cup \{j\}} \cost(E(P_{i'})) + c_{uv} 
                \leq \sum_{i'\in J \cup \{j\}} (\Delta -
                                   \epsilon_{i'})+c_{uv} 
                                   \leq (|J|+1) \Delta \leq \frac{|J|+1}{\ell} \cdot OPT_{LP}. \]
The last bound follows because the feasible dual we have grown has value $\ell \cdot \Delta \leq OPT_{LP}$.
Let $\ell' = \ell - |J|$ be the number of nonroot components in $\pt'$.
Conclude by observing $|J|+1 = \ell - \ell' + 1 \leq 2(\ell - \ell')$.
\end{proof}

To wrap things up, executing Algorithm \ref{alg:grow} and constructing the partial Steiner tree as in Lemma \ref{lem:augment} yields
the partial Steiner tree that is promised by Lemma \ref{lem:grow}.


\section{Conclusion}

We have shown that the integrality gap of LP relaxation (\ref{lp:primal}) is $O(\log k)$ in quasi-bipartite instances of directed Steiner tree.
The gap is known to be $\Omega(\sqrt k)$ in 4-layered instances \cite{zosin:khuller} and $O(\log k)$ in 3-layered instances \cite{friggstad+14}.
Since quasi-bipartite graphs are a generalization 2-layered instances, it is natural to ask if there is a generalization of 3-layered instances
which has an $O(\log k)$ or even $o(\sqrt k)$ integrality gap.

One possible generalization of 3-layered graphs would be when
the subgraph of $G$ induced by the Steiner nodes does not have a node with both positive indegree and positive outdegree. 
None of the known results on directed Steiner tree suggest such instances have a bad gap.

Even when restricted to 3-layered graphs, a straightforward adaptation of our algorithm
that grow moats around the partial Steiner tree heads until some partial Steiner trees absorbs another
fails to grow a sufficiently large dual to pay for the augmentation within any reasonable factor. A new idea is needed.


%









\end{document}